\documentclass[prl,showpacs,preprint,amsmath,aps]{revtex4-1}
\usepackage{epsfig,graphicx}
\usepackage{amsmath}
\usepackage{amssymb}
\usepackage{mathrsfs}
\usepackage{url}
\usepackage{amsthm}
\usepackage{bm}
\usepackage{hyperref}
\usepackage{amsmath, bm}
\usepackage{float}
\usepackage{epstopdf}

\hypersetup{colorlinks=true, citecolor=blue, urlcolor=blue, linkcolor=blue}

\newtheorem{theorem}{Theorem}
\newtheorem{lemma}{Lemma}

\newtheorem{exmp}{Example}

\begin{document}
\title{ Device-independent randomness based on a tight upper bound of the maximal quantum value of chained inequality}
\author{Youwang Xiao$^1$}
\author{XinHui Li$^2$}
\author{Jing Wang$^1$}
\author{Ming Li$^{1}$\footnote{E-mail: liming@upc.edu.cn.}}
\author{Shao-Ming Fei$^3$}
\affiliation{$^1$College of the Science, China University of
Petroleum, 266580, Qingdao, China\\
$^2$ National Laboratory of Solid State Microstructures, School of Physics,
and Collaborative Innovation Center of Advanced Microstructure, Nanjing University, 210093, Nanjing, China\\
$^3$School of Mathematical Science, Capital Normal University, 100048, Beijing, China}
\date{\today}
\begin{abstract}
The violation of Bell inequality not only provides the most radical departure of quantum theory from classical concepts, but also paves the way of applications in such as device independent randomness certification. Here, we derive the tight upper bound of the maximum quantum value for chained Bell inequality with arbitrary number of measurements on each party. The constraints where the upper bound saturates are also presented. This method provides us the necessary and sufficient conditions for some quantum states to violate the chained Bell inequality with arbitrary number of measurements. Based on the tight upper bound we present the lower bounds on the device independent randomness with respect to the Werner states. In particular, we present lower bounds on the randomness generation rates of chained Bell inequality for different number of measurements, which are compared with the family of Bell inequalities proposed by Wooltorton et al. [Phys. Rev. Lett. 129, 150403 (2022)].  Our results show that chained Bell inequality with three measurements has certain advantages at a low level of noise and could be used to improve randomness generation rates in practice.
\end{abstract}

\pacs{03.67.Mn,03.65.Ud}

\keywords{Chained Bell inequality; Tight upper bound; Device-independent randomness.}

\maketitle
\smallskip
\section{\leftline{1. Introduction}}
Bell inequalities must be satisfied by any local and realistic theory \cite{ref1}. The impossibility of local hidden variable (LHV) theories in quantum mechanics is witnessed by the violation of a Bell inequality. The Bell inequalities play a central role for device independent quantum information tasks such as quantum key distribution \cite{ref2,ref3,ref4}, randomness certification \cite{ref5,ref6}, self-testing \cite{ref7,ref8} and communication complexity \cite{ref9}.

Bell inequalities can be characterized by the number of parties $k$, the number of measurement settings $n$, and the number of possible outcomes for each measurement $m$. The well-known Clauser-Horne-Shimony-Holt (CHSH) inequality \cite{ref10} is for the case $k=n=m=2$. The chained Bell inequalities \cite{ref11} are the generalization of the CHSH inequality.
Interestingly, the non-locality and entanglement may behavior differently in certain scenarios \cite{ref12} in terms of the chained Bell inequalities. The chained Bell inequalities also give rise to the nonlocal restriction on statistics of measurement outcomes \cite{ref13}. Self-testing protocols have been presented in \cite{ref14} based on the chained Bell inequalities. It has been shown that the quantum state and the measurements which give rise to the maximal violation of the chained Bell inequalities are unique. The chained Bell inequalities have been also implemented to observe the violation experimentally for $n\geq 3$ without detection loophole \cite{ref15}.

Randomness is an essential resource in many application from numerical simulations to cryptography. The violation of Bell inequalities can certify private randomness that is guaranteed to be uncorrelated to any outside process or variable in a device-independent(DI) way \cite{ref5}. In last few years, DI technologies have been studied intensively and device-independent quantum random-number generator(DI QRNG) has attracted huge attention among them, which is able to access randomness by observing the violation of Bell inequalities without any assumptions about the source and measurement device \cite{ref5,ref16}. The randomness of the
output pairs conditioned on the input pairs for the standard  nonlocality-based scenario can be quantified by the min-entropy, where upper bounds on the guessing probabilities can be computed using the Navascu\'{e}s-Pironio-Ac\'{i}n(NPA) hierarchy for quantum
correlations \cite{ref18}. Thus, for any given Bell inequality one can obtain  lower bounds on the min-entropy  as a function of the amount of Bell inequality violation. However, the complete measurement statistics produced in a Bell experiment can offer much more information than merely resorting to the degree of violation of a single Bell inequality \cite{ref19,ref20}. It is worth pointing out that the DI randomness bounds using specific Bell operators are always lower than using the complete measurement statistics\cite{ref21}. Recently, a simple adaptation of the original NPA hierarchy for sequential Bell scenarios is provided and one can robustly certify over 2.3 bits of device-independent local randomness from a two-quibt state using a sequence of measurements, going beyond the theoretical maximum of two bits that can be achieved with non-sequential measurements\cite{ref22}.

For a given Bell inequality, it is of significance to estimate the maximal violation for given quantum states. In this paper, we derive a tight upper bound on the maximal quantum violation of the chained Bell inequality with arbitrary number of measurements for two-qubit systems. Then, we provide the constraints on the quantum state for the tightness of the upper bound. In addition, we present the optimal violation saturating the upper bound for several quantum states, including the Werner state. Then, based on the tight upper bound we present the lower bounds on the device independent randomness with respect to the Werner states. Moreover, we explore the robustness of chained Bell inequality under a Werner state noise model and compare it  to that of the family of Bell inequalities $J_{\gamma }$ proposed in Ref.\cite{ref6}. The results show that chained Bell inequality with $n=3$ has a better performance for a low noise and is useful for practical DI randomness generation.

\section{2. Tight upper bound on the maximal quantum value of chained inequality}
Let $A_{i}=\sum_{k=1}^{3}a_{ik}\sigma _{k}$ and $B_{j}=\sum_{l=1}^{3}b_{jl}\sigma _{l}$
be Hermitian operators with eigenvalues $\pm 1$, $i,j=1,...,n$, where
$\sigma _{1}=\bigl(\begin{smallmatrix}
0 &1 \\
 1& 0
\end{smallmatrix}\bigr),~\sigma _{2}=\bigl(\begin{smallmatrix}
0 &-i\\
 i& 0
\end{smallmatrix}\bigr),~\sigma _{3}=\bigl(\begin{smallmatrix}
1 &0\\
 0& -1
\end{smallmatrix}\bigr) $ are the standard Pauli matrices.
Let  $B_{C^{n}}$ denote the Bell operator for the chained Bell inequality with $n$ measurements on each party,
\begin{equation}\label{bicon1}
B_{C^{n}}=\sum_{k=1}^{n-1}A_{k}\otimes B_{k}+A_{k+1}\otimes B_{k}+A_{n}\otimes B_{n}-A_{1}\otimes B_{n}
\end{equation}
and $C^{n}$ denote the expectation value of the Bell operator,
\begin{equation}%\label{SI}
C^{n}=\sum_{k=1}^{n-1}\left \langle A_{k}B_{k} \right \rangle+\left \langle A_{k+1}B_{k} \right \rangle+\left \langle A_{n}B_{n} \right \rangle-\left \langle A_{1}B_{n} \right \rangle,
\end{equation}
where $\left \langle A_{x}B_{y} \right \rangle=\sum _{ab}abp(a,b|x,y)$ is the expectation value of the measurement outcome $(a,b)\in \left \{ +1,-1 \right \}$ with respect to the inputs $x$ and $y$. The joint probabilities $p(ab|x,y)$ can be written as $p(ab|x,y)=\frac{1}{4}(1+a\left \langle A_{x} \right \rangle+b\left \langle B_{y} \right \rangle+ab\left \langle A_{x}B_{y} \right \rangle)$ in terms of the expectation values of the measurements $A_{x}$ and $B_{y}$.
The classical bound for LHV models and the maximum quantum value of $C^{n}$ are respectively
given by Ref. \cite{ref23},
\begin{equation}
C_{LHV}^{n}=2n-2,
\end{equation}
\begin{equation}
C_{Q}^{n}=2ncos\frac{\pi }{2n}.
\end{equation}

Before estimating the violation of the chained Bell inequality, we first prove the following lemma.

\begin{lemma}%\label{lemma.}
Let $A$ be an $n\times n$ matrix and $A=U^{T}\Sigma V$ the singular value decomposition of $A$. Suppose the column vectors of the matrices $U$ and $V$ are $u_{1},...,u_{n}$ and $v_{1},...,v_{n}$ respectively. For any vectors $x=\sum x_{i}u_{i}$ and $y=\sum y_{i}v_{i}$, we have
\begin{equation}%\label{SI}
\left | x^{T}Ay \right |\leq \sigma _{max}\left \| x \right \|\left \| y \right \|,%²»ÄÜÓÃ\leqslant
\end{equation}
where $\sigma _{max}$ is the largest singular value of the matrix $A$ and $\left \| \cdot  \right \|$ stands for $2$-norm. The equality holds if and only if $u_{i}$ and $v_{i}$ are the left and right singular vectors corresponding to the maximum singular value, $x_{i}\cdot y_{i}$ keep the same symbol, and $x_{i}$ and $y_{i}$ are proportional for $i=1,...,n$.
\end{lemma}

\begin{proof}
According to singular value decomposition of $A$, we have
\begin{align}\label{proof of Lemma}
%\left | x^{T}Ay \right |=\left | (\sum x_{i}u_{i})^{T}(\sum y_{j}v_{j})
                        % \right |\le\sum_{ij}\left | x_{i}y_{j}u_{i}^{T}Av_{j} \right |=\sum_{i}\left | x_{i}y_{i}\sigma _{i} \right |\leq \sigma _{max}\sum_{i}\left | x_{i}y_{i} \right |\leq \sigma _{max}\left \| x \right \|\left \| y \right \|
\notag\left | x^{T}Ay \right |=\left | (\sum x_{i}u_{i})^{T}A(\sum y_{j}v_{j})
                        \right |&\le\sum_{ij}\left | x_{i}y_{j}u_{i}^{T}Av_{j} \right |=\sum_{i}\left | x_{i}y_{i}\sigma _{i} \right |\\
\notag&\leq \sigma _{max}\sum_{i}\left | x_{i}y_{i} \right |\\
\notag&\leq \sigma _{max}\left \| x \right \|\left \| y \right \|.
\end{align}
%{\color{blue}The first inequality holds if $x_{i}\cdot y_{i}$ keep the same symbol for $i=1,...,n$. The second inequality holds if $u_{i}$ and $v_{i}$ are the left and right singular vectors corresponding to the maximum singular value. The third inequality holds if $x_{i}$ and $y_{i}$ are proportional for $i=1,...,n$.}
The first inequality holds if $x_{i}\cdot y_{i}$ keep the same symbol for $i=1,...,n$. The second inequality holds if $u_{i}$ and $v_{i}$ are the left and right singular vectors corresponding to the maximum singular value. The third inequality holds if $x_{i}$ and $y_{i}$ are proportional for $i=1,...,n$.
\end{proof}

An arbitrary two-qubit state $\rho$ has the following Bloch representation,
\begin{equation}
\rho =\frac{1}{4}\left [ I\otimes I+\sum_{i}r_{i}\sigma _{i}\otimes I+I\otimes \sum_{i}s_{j}\sigma _{j}+\sum_{i,j}m_{ij}\sigma _{i}\otimes \sigma _{j} \right ],
\end{equation}
where $\boldsymbol r=(r_1, r_2, r_3)$ and $\boldsymbol s=(s_1, s_2, s_3)$ are three-dimensional real vectors and $M=(m_{ij})$ is the real correlation matrix.
The following theorem holds for any states $\rho$.

\begin{theorem}
The maximum quantum value {$Q_{C^{n}}(\rho )$} of the chained Bell inequality satisfies
\begin{equation}\label{thm1}
 {Q_{C^{n}}(\rho )\equiv \max_{ A_{i},B_{j}}\left | C^{n} (\rho) \right |\leq 2ncos\frac{\pi }{2n}\sigma _{max},}
\end{equation}
where  {$C^{n} (\rho) =Tr(B_{C^{n}}\rho )$}, the maximum is taken over all possible  { observables $A_{i},B_{j}$}, $\sigma _{max}$ is the maximum singular value of the correlation matrix $M=(m_{kl})$ with entries $m_{kl}=tr(\rho (\sigma _{k}\otimes\sigma _{l} ))$, $k,l=1,2,3$.
\end{theorem}
\begin{proof}
First, note that
\begin{align}
\notag\left \langle A_{i}B_{j} \right \rangle_{\rho }&=tr(A_{i}\otimes B_{j}\rho )\\
\notag&=\sum_{k,l=1}^{3}a_{ik}b_{jl}m_{kl}\\
&=a_{i}^{T}Mb_{j}.
\end{align}
It follows from Lemma $1$ and Cauchy-Schwartz inequality that
\begin{align}\label{bicon2}
\notag { C^{n}(\rho)}&=a_{1}^{T}M(b_{1}-b_{n})+\sum_{k=1}^{n-1}a_{k+1}^{T}M(b_{k}+b_{k+1})\\
\notag&\leq \sigma _{max}\left [ \left | b_{1}-b_{n} \right |+\left |   b_{1}+b_{2}\right |+ ...+\left | b_{n-1}+b_{n}  \right |\right ]\\
&\leq \sigma _{max}\sqrt{n}\sqrt{2n+2(\left \langle b_{1},b_{2} \right \rangle+...+\left \langle b_{n-1},b_{n} \right \rangle-\left \langle b_{1},b_{n} \right \rangle)}.
\end{align}
To estimate the maximum value of the last formula above, we consider the following optimization problem: $\max_{\left \{ b_{i}\in R^{n} \right \}}(\left \langle b_{1},b_{2} \right \rangle+...+\left \langle b_{n-1},b_{n} \right \rangle-\left \langle b_{1},b_{n} \right \rangle)$.

Let $M=[m_{ij}]$ be the Gram matrix of the vectors $\left \{ b_{1},b_{2},...,b_{n} \right \}\subseteq R^{n}$ with respect to the inner product:
 \begin{equation}
 M=\begin{pmatrix}
\left \langle b_{1},b_{1} \right \rangle & \left \langle b_{1},b_{2} \right \rangle & \cdots  & \left \langle b_{1},b_{n} \right \rangle\\
 \left \langle b_{2},b_{1} \right \rangle& \left \langle b_{2},b_{2} \right \rangle & \cdots  &\left \langle b_{2},b_{n} \right \rangle \\
 \vdots &\vdots   & \cdots  & \vdots \\
\left \langle b_{n},b_{1} \right \rangle & \left \langle b_{n},b_{2} \right \rangle & \cdots  & \left \langle b_{n},b_{n} \right \rangle
\end{pmatrix},
 \end{equation}
where all $b_{i}$ are unit vectors and $M\succeq 0$. Define
 \begin{equation}
 W=\begin{pmatrix}
0 & 1 &  0 & \cdots  & 0 & -1\\
1 & 0 &  1 & \cdots  & 0 & 0\\
0 & 1 &  0 & \cdots  & 0 & 0\\
\vdots  & \vdots  & \vdots  & \cdots  & \vdots  & \vdots \\
0 & 0 &  0 & \cdots  & 0 & 1\\
-1 & 0 &  0 & \cdots  & 1 & 0
\end{pmatrix}.
 \end{equation}
We can rephrase the optimization problem as the following semi-definite programm (SDP):
 \begin{equation}
 \begin{aligned}
  & {\rm maximize}\quad \frac{1}{2}tr(MW)\\
  & {\rm subject~ to}\quad M\succeq 0,~m_{ii}=1~ \forall i.
  \end{aligned}
 \end{equation}
We then obtain the dual formulation of SDP
 \begin{equation}
 \begin{aligned}
  & {\rm minimize}\quad tr(diag(v))\\
  &{\rm subject~ to}\quad -\frac{1}{2}W+diag(v)\succeq 0.
  \end{aligned}
 \end{equation}

%Let $p^{*}$ and $d^{*}$ denote the optimal values for the primal and dual problems above respectively. According to the weak duality, we have $d^{*}\geq p^{*}$. Actually, we can prove $d^{*}=p^{*}$ for the problem.

We choose unit vectors $b_{i}$ to be of the form \cite{ref23},
\begin{equation}
b_{i}=(cos(\varphi _{i}),sin(\varphi _{i}),0,\cdots ,0),
\end{equation}
where $\varphi_{i}=\frac{\pi }{2n}(2i-1))$. As $\varphi _{i+1} - \varphi _{i}=\frac{\pi }{n}$ we see that $\left \langle b_{i},b_{i+1} \right \rangle= {cos(\frac{\pi }{n})}$, $1\leq i\leq n-1$. The angle between $-b_{1}$ and $b_{n}$ is $\pi -\varphi _{n}=\frac{\pi }{n}$. Thus $-\left \langle b_{1},b_{n} \right \rangle=cos(\frac{\pi }{n})$. As a result, the value of primal problem is given by
\begin{equation}
p^{'}=ncos(\frac{\pi }{n}).
\end{equation}
The corresponding matrix ${M}'$ given the above vectors obviously satisfies the constraints. Hence, $p^{'}$ is a feasible solution of the primal problem.

Next taking the $n$-dimensional vector
\begin{equation}
{v}'=cos(\frac{\pi }{n})(1,\cdots ,1),
\end{equation}
we obtain the maximum eigenvalue $2cos(\frac{\pi }{n})$ of the matrix $W$. Noting that $W$ and $diag({v}')$ are both Hermitian, we have
\begin{equation}
 \begin{aligned}
  \lambda _{min}(-\frac{1}{2}W+diag({v}'))&\geq \lambda _{min}(-\frac{1}{2}W)+\lambda _{min}(diag({v}'))\\
  &=-\frac{1}{2}(2cos(\frac{\pi }{n}))+cos(\frac{\pi }{n})\\
  &=0,
  \end{aligned}
 \end{equation}
where $\lambda _{min}(M)$ is the minimum eigenvalue of matrix $M$. Thus $-\frac{1}{2}W+diag({v}')\succeq 0$ and $diag({v}')$ is a feasible solution to the dual problem. The value of the dual problem is
\begin{equation}
{d}'=ncos(\frac{\pi }{n}).
\end{equation}
Since ${d}'={p}'$, it is of strong duality, which implies that ${M}'$ and ${v}'$ are the optimal solutions of the primal and dual problems, respectively. Finally, considering $b_{i}$ as three-dimensional unit vectors, from formula \eqref{bicon2} we prove the Theorem 1.
\end{proof}

 {\emph{Tightness of the upper bound}. Let us now discuss the conditions where upper bound \eqref{thm1} saturates. We can derive from lemma 1 that the first inequality in Eq.\eqref{bicon2} saturates if $b_{1}-b_{n}, b_{1}+b_{2},...,b_{n-1}+b_{n}$ are the singular vectors corresponding to the maximum singular value $\sigma _{max}$ and $a_{1}=\frac{M(b_{1}-b_{n})}{|M(b_{1}-b_{n})|},a_{2}=\frac{M(b_{1}+b_{2})}{|M(b_{1}+b_{2})|},...,a_{n}=\frac{M(b_{n-1}+b_{n})}{|M(b_{n-1}+b_{n})|}$. According to Cauchy-Schwartz inequality, the second inequality in Eq.\eqref{bicon2} is saturated if $\left | b_{1}-b_{n} \right |=\left |   b_{1}+b_{2}\right |=...=\left | b_{n-1}+b_{n}  \right |$. In addition, condition $\left \langle b_{1},b_{2} \right \rangle+...+\left \langle b_{n-1},b_{n} \right \rangle-\left \langle b_{1},b_{n} \right \rangle=ncos\frac{\pi }{n}$ needs to be satisfied in order to reach the upper bound. By analyzing the conditions above, as long as the degeneracy of $\sigma _{max}$ of the correlation matrix for quantum state $\rho$ is more than one, the upper bound is always tight.}
\begin{exmp}
Consider the singlet state,
\begin{equation}\label{2}
|\psi _{-}\rangle=\frac{1}{\sqrt{2}}(|01\rangle-|10\rangle {)}.
\end{equation}
Choose the measurements for Alice and Bob given by
\begin{equation}\label{3}
a_{i}=(-sin\frac{i-1}{n}\pi ,0,-cos\frac{i-1}{n}\pi ),
\end{equation}
\begin{equation}\label{4}
b_{j}=(sin\frac{2j-1}{2n}\pi ,0,cos\frac{2j-1}{2n}\pi ),
\end{equation}
$i,j=1,...,n$, respectively. We obtain the maximum quantum non-locality for arbitrary number  $n$ of measurements,
\begin{equation}
 {Q_{C^{n}}(|\psi _{-}\rangle)}=2ncos\frac{\pi }{2n}.
\end{equation}
It can be verified that the largest singular value of the matrix $M$ corresponding to state $|\psi _{-}\rangle$ is 1, which shows that the upper bound in \eqref{thm1} is tight.
\end{exmp}

\begin{exmp}
Consider the Werner state,
\begin{equation}\label{5}
\rho (p)=p |\psi _{-}\rangle\langle\psi _{-}|+(1-p)\frac{I_{4}}{4},
\end{equation}
where $0\leq p\leq 1$, and $I_{4}$ is the identity matrix on $\mathbb{C}^{2}\otimes \mathbb{C}^{2}$. The correlation matrix of $\rho(p)$ is
$M=diag(-p,-p,-p)$. By choosing the measurement settings \eqref{3} and \eqref{4}, the upper bound also remains tight for the Werner state, namely,
\begin{equation}\label{E2}
 {Q_{C^{n}}(\rho (p))}=2npcos\frac{\pi }{2n}.
\end{equation}
\end{exmp}
The range of entanglement witnessed by the chained Bell inequality for the Werner state is $p\geq \frac{n-1}{ncos\frac{\pi }{2n}}(n\geq 2)$.
% {\begin{exmp}
%Consider the following two-qubit states\cite{ref25} of the form
%\begin{align}\label{E31}
%\rho(q)=q|\psi _{-}\rangle\langle\psi _{-}|+(1-q)|0\rangle\langle0|\otimes \frac{I_{2}}{2}
%\end{align}
%where $0\leq q\leq 1$, and $I_{2}$ is the identity matrix on $\mathbb{C}^{2}$. The correlation matrix of $\rho(q)$ is
%$M=diag(-q,-q,-q)$. By choosing the measurement settings \eqref{3} and \eqref{4}, the upper bound also remains tight for the state\eqref{E31}, namely,
%\begin{equation}\label{E32}
%Q_{C^{n}}(\rho (q))=2ncos\frac{\pi }{2n}q.
%\end{equation}
%\end{exmp}
%It has been proven that the state $\rho(q)$ admits a local model
%for projective measurements when $q\leq \frac{1}{2}$, although it is entangled for all $q> 0$\cite{ref25}. Significantly, the particular quantum state $\rho(q)$ that can originally be described using a local hidden variable model surprisingly exhibit  hidden nonlocality by employing judicious local filters before a standard Bell test is performed\cite{ref25}. Moreover, this phenomenon has been demonstrated experimentally that the violation of the CHSH Bell inequality is observed by applying local filters\cite{ref26}. Therefore, it is possible to observe the violation of the chained Bell inequality for arbitrary measurements per party by applying  local filters to the state \eqref{E31} with hidden nonlocality and  utilizing the expression  \eqref{E32}. Notably, it was shown that hidden nonlocality is not a general feature\cite{ref27} and, for example, the two-qubit Werner state\eqref{5} remains local even after arbitrary local filters.
 {\begin{exmp}
Consider the following special X-states
\begin{align}
\rho(\nu ,l) =\frac{1}{4}\left [ I_{2}\otimes I_{2}+\nu \sigma _{x}\otimes \sigma _{y}+\nu \sigma _{y}\otimes \sigma _{x}+l \sigma _{z}\otimes \sigma _{z} \right ]
\end{align}
The state $\rho(\nu ,l)$ will be entangled  if
\begin{align}
1-2\nu +l> 0,1-2\nu -l< 0,0< l< 1
\end{align}
The singular value of the correlation matrix $$M=\begin{pmatrix}
0 & \nu  & 0\\
\nu  & 0 & 0\\
 0& 0 & l
\end{pmatrix}$$ are $\nu,\nu$ and $l$. Thus, adding another condition $\nu \geq l$, we can obtain  the optimal violation $2ncos\frac{\pi }{2n}\nu $ of $\rho(\nu ,l)$ for chained Bell inequality $C^{n}$.
\end{exmp}
 To show that the upper bound is tight and for the sake of simplicity, we can choose $\nu=\frac{3}{4}l+\frac{1}{4},\frac{2}{5}<\nu <1$ and $n=3$ and use particle swarm optimization algorithm to get numerically true maximum. The results are shown in Fig. 1. Except for additional points (due to the accuracy of the algorithm), the numerical results are basically consistent with the upper bound.}

\begin{figure}[http]
 \centering
 \setlength{\abovecaptionskip}{0.2cm}
  \includegraphics[width=11cm,height=8.5cm]{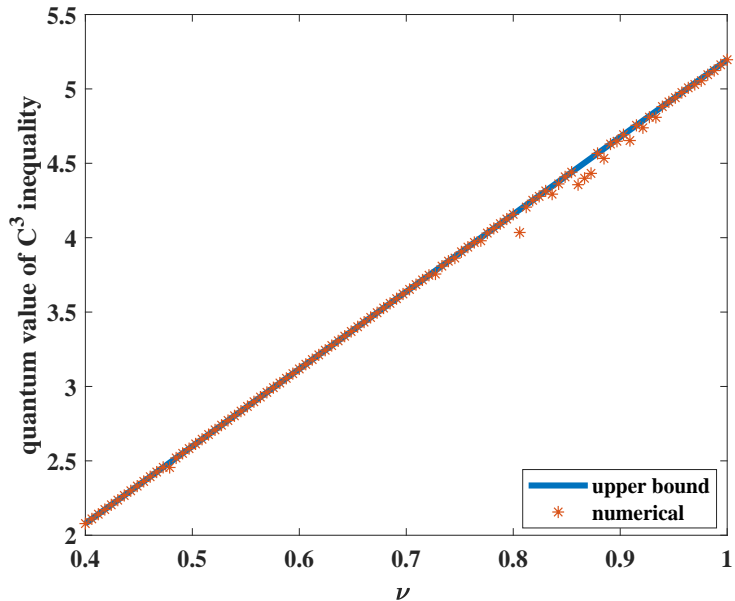}
  \caption{Comparison of numerically true maximum and upper bound of quantum value of $C^{3}$ inequality  for  X-states $\rho(\nu ,l)$ with $\nu$ from 0.4 to $1$.}
\end{figure}

\section{3. Randomness certification based on chained Bell inequality}

 { Quantum phenomena display intrinsic randomness and can thus serve as QRNG \cite{ref24}. Motivated by the device-independent techniques,  Bell inequality has been studied as a useful method to various quantum information processing tasks, such as DIQRNG which enables one to access true randomness without any assumptions on the source and measurement devices \cite{ref5}. }

Given a quantum distribution $P$, the randomness of the output pairs $(a,b)$ conditioned on the input pairs $(x,y)$ for the entangled pairs can be quantified by the min-entropy, which is defined as $H_{\infty}(A,B|x,y)=-log_{2}{\rm max}_{ab}\,p(ab|xy)$. To get  the min-entropy in device independent scenario, one needs to maximize the $p(ab|xy)$ running over all possible  quantum realizations $\{\rho,M_x^a,M_y^b\}$ compatible with $P$.  {However, optimization over the set of quantum realizations is computationally hard. To  resolve this technical difficulty,  the most of previous work on DIQRNG exploit the lower bound of min-entropy by the numerical method where combing NPA hierarchy and  semidefinite programming (SDP) \cite{ref18},  which was first proposed in Ref.\cite{ref5}.}

 {In this work, we explore the lower bounds for DIQRNG with two different kinds of constraint conditions on the SDP, they are the Bell inequality violation \cite{ref5} and the complete measurement statistics \cite{ref19,ref25}. }First, for the any observed violation $I$ of a given Bell inequality,
 we denote the maximum guessing probability $P^{\ast }(ab|xy)$ as the solution to
the following optimization problem:
\begin{align}
P^{\ast }(ab|xy)&=\max_{\left \{ \rho ,M_{x}^{a},M_{y}^{b}  \right \}}p(ab|xy)\\
  {\rm subject~ to }&\sum_{abxy}c_{abxy}P(ab|xy)=I\\
  &p(ab|xy)=tr(\rho M_{x}^{a}\otimes M_{y}^{b})
  \end{align}
  where the optimization is carried over all states $\rho$ and all measurement  operators $M_{x}^{a}$ and $M_{y}^{b}$, defined over Hilbert spaces of arbitrary dimension.
Resorting to SDP techniques with the NPA hierarchy, we can numerically bound the randomness based on  the amount of the Bell violation. That is, for any given Bell inequality, one can derive
bounds $H_\infty(A,B|x,y)\geq f(I)$ on the min-entropy.

Then, the DI global randomness based on the complete measurement statistics is defined by
 \begin{align} P_{\text{guess}}(A,B|E,x,y)&=\max_{\left \{  |\varphi \rangle,M_{a}^{x},M_{b}^{y},M_{e}\right \}}\sum_{ab}\left \langle \varphi |M_{a}^{x}\otimes  M _{b}^{y}\otimes M _{e=ab}|\varphi \right \rangle,\\
  \text{subject to } &p(ab|xy)=\left \langle \varphi |M _{a}^{x}\otimes  M _{b}^{y}\otimes \mathbb{I}|\varphi \right \rangle
  \end{align}
where the relation between randomness and non-locality can be considered from the perspective of nonlocal guessing games and the quantity gives the maximum probability that an adversary Eve's outcome $e$ matches the outcome $(a,b)$ for measurement choice $(x,y)$ over all possible quantum realizations, described by a tripartite quantum state $|\varphi  \rangle$ and measurements $M_{a}^{x}, M_{b}^{y}, M _{e}$ for devices $A,B,E$, compatible with the observed full joint probabilities distribution $P(ab|xy)$. The guessing probability $P_{\text{guess}}(A,B|E,x,y)$ is directly connected to the min-entropy $H_{\infty}(A,B|E,x,y)=-\log_{2}P_{\text{guess}}(A,B|E,x,y)$ representing the bits of randomness certified in this Bell scenario.

In a Bell setup for generating new random numbers, the user of a DIQRNG has no knowledge about the mechanisms of the device. Due to noise the quantum state used to generate random numbers may be not always maximally entangled.  We consider the randomness that can be extracted from complete measurement statistics maximally violating $C^{3}$ inequality in the presence of white noise. The correlations are of the form \begin{align}\label{ConsNPA}
\bm{p}=p\bm{q}+(1-p)\bm{r},
\end{align}
where $\bm{p}$ are the quantum correlations yielding the maximal value $3\sqrt{3}p$ from the Werner state, $\bm{q}$ are the
quantum correlations yielding the maximal $C^{3}$ inequality violation of $3\sqrt{3}$ and $\bm{r}$ denotes completely random correlations of the form $p(ab|xy)=\frac{1}{4}$ for all $a,b,x$ and $y$.  {We optimize the lower bounds of the global randomness based on the correlations defined as Eq.\eqref{ConsNPA} on Q1 and Q2 levels of NPA matrix. As shown in Fig. \ref{Rand1}, for a given noise $p$, the extractable randomness  using
the complete measurement statistics is more than that from an optimal violation of the $C^3$
inequality, both for Q1 and Q2 levels. It can be explained as the constraint with complete measurement statistics
is much tighter than knowledge of the $C^3$ violation alone, which  results in a higher lower bound of randomness.}  { Furthermore, it can be shown from Fig. $2$ that the randomness can be extracted in the $C^{3}$ scenario only provided that the $C^{3}$ inequality is violated, i.e., $p>4/(3\sqrt{3})\approx 0.7698$ and $1.1$ bits of randomness can be certified from the maximal violation of the $C^3$ inequality under  measurements setting $A_1B_1$ when $p=1$ as the worst case of the extractable randomness from Q2 level.}

\begin{figure}[ht]
 \centering
 \includegraphics[scale=0.7]{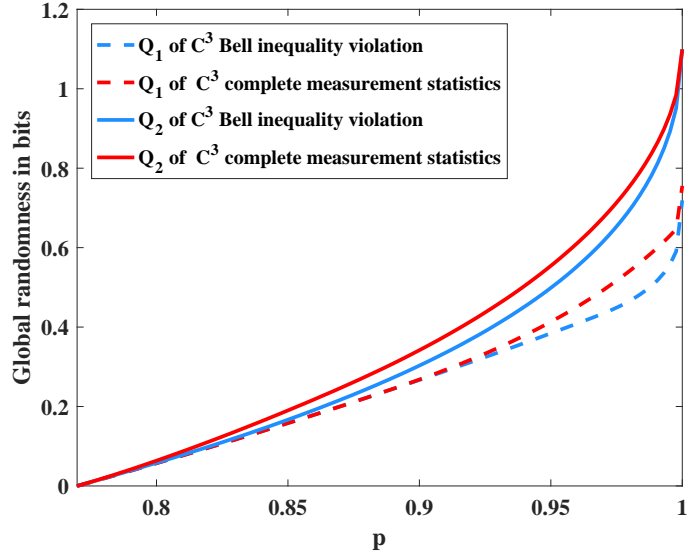}
  \caption{ {Lower bounds of DI global randomness under  measurement setting $A_1B_1$ via different constraint conditions. For a given $p$, the red color curves are obtained with the complete measurement statistics running on Q1 (dashed) and Q2 (solid) levels of NPA hierarchy; while the blue color curves are obtained by the   optimal  $C^{3}$ inequality violation $3\sqrt{3}p$ running on Q1 (dashed) and Q2 (solid) levels
of NPA hierarchy.}\label{Rand1}}
\end{figure}

Additionally, the self-testing protocols for  the chained Bell inequalities have been presented in Ref.\cite{ref14}, where the authors proved that the maximal violation of the chained Bell inequality uniquely certifies the state \eqref{2}  and the measurements \eqref{3}-\eqref{4}  up to a local isometry.
For joint measurement settings appearing in the chained Bell inequality, we have the following expectation values
\begin{equation}
\begin{aligned}
&\left \langle A_{k}B_{k} \right \rangle=\left \langle A_{k+1}B_{k} \right \rangle=cos(\frac{1}{2n}\pi ),~~k=1,...,n-1,\\
&\left \langle A_{n}B_{n} \right \rangle=-\left \langle A_{1}B_{n} \right \rangle=cos(\frac{1}{2n}\pi ).
\end{aligned}
\end{equation}
at the point of maximum violation. Hence, the maximum guessing probability for above settings are given by $(1+cos(\pi /2n))/4$ which is in consistent with  {the maximum global randomness of $1.1$ bits certified by $C^{3}$ inequality ($n=3$) under  measurements setting $A_1B_1$}. In particular, for an odd number of measurements, the expected values of the $n$ correlators that are not appeared in the chained Bell inequality, with respect to the measurement settings $(1,\frac{n+1}{2}),(2,\frac{n+3}{2}),\cdots,(\frac{n+1}{2},n),(\frac{n+3}{2},1),\cdots ,(n,\frac{n-1}{2})$, equal to $0$, which certify two bits of maximal randomness.

 {To investigate the quantum randomness generated by the quantum correlation, we explore the amount of randomness extracted in the chained Bell scenario with the measurement settings that certify maximum randomness. By adding the same noise defined in Eq.\eqref{ConsNPA} to the ideal quantum violation of Bell inequality, the lower bounds of randomness that can be extracted from the corresponding measurements settings for different noise $p$ are shown in Fig.\ref{Rand2}. We solve the SDP on the $1+AB$ level of NPA matrix for different values of $n$. CHSH Bell inequality enables one to extract the randomness from visibility $p\geq 1/\sqrt{2}\approx0.7071$, which is able to certify randomness from a more wider visibility range of $p$ than chained Bell inequality with $n\geq 3$. Since the optimal violation of the Werner state is $2npcos(\pi /2n)$, we can see that the range of randomness certified decreases as the number of measurements increases. While for a given $p$ approaching to 1, i.e., the source state closes to singlet state, the lower bounds of randomness for chained Bell inequality are higher than  CHSH inequality, which shows that chained inequalities have a better performance in randomness generation rates. In particular, we compare the robustness of chained Bell inequality with  the family of Bell inequalities $J_{\gamma }$ with $0\leq \gamma \leq \pi /12$ proposed by Wooltorton et al. \cite{ref6},
 \begin{align}
J_{\gamma }=\left \langle A_{0}B_{0} \right \rangle+(4cos^{2}\left [ \gamma +(\pi /6) \right ]-1)(\left \langle A_{0}B_{1} \right \rangle+\left \langle A_{1}B_{0} \right \rangle-\left \langle A_{1}B_{1} \right \rangle)
\end{align}
 where, for simplicity, we use Werner state and noiseless measurements relying on the choice of self-testing strategy and the randomness is optimized over different parameters $\gamma$ at each level of noise. As can be seen from Fig. \ref{Rand2}, there is a slight improvement in the DI randomness bound with three measurements at a low level of noise.}

\begin{figure}[http]
\label{figure2}
 \centering
 \setlength{\abovecaptionskip}{0.2cm}
\includegraphics[width=11cm,height=8.5cm]{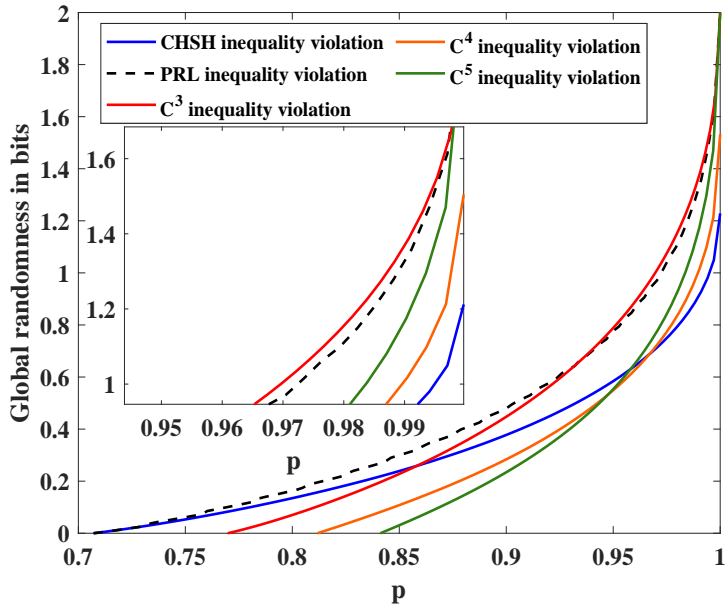}
  \caption{ {The comparisons between chained Bell inequality and previous work in term of lower bounds on DI randomness under corresponding measurement setting that certifies maximum randomness. The solid curves are obtained by the CHSH inequality violation (blue) and the chained inequality violation with $n=3$ (red),$n=4$ (orange) and $n=5$ (green). The dashed curves are obtained by the violation of the family of Bell inequalities $J_{\gamma }$, where  the randomness is optimized over different parameters $\gamma$ at each level of noise. All the curves were obtained using  1+AB level of NPA hierarchy. The inset is to  provide a clearer view of the differences between these curves.}}\label{Rand2}
\end{figure}

\section{4. Conclusions and Discussions}
In this work, we have introduced a method to obtain the tight upper bound of the maximum quantum value for chained Bell inequality with arbitrary measurement settings as a generalization of the CHSH Bell inequality. {We have investigated the constraints on the quantum state for the tightness of the upper bound and provided strong numerical evidence to support our claim.} In particular, the entanglement of the Werner state witnessed by the chained Bell inequality has been demonstrated. Bell inequalities have found applications not only in foundations of physics but also in protocols of DI quantum information processing tasks \cite{ref4} where the efficiency of the hardware is proportional to the violation of a Bell inequality, while the resource being used is entanglement. Thus, based on the tight upper bound we present the lower bounds on the device independent randomness with respect to the Werner states by using two methods of chained inequality optimal violation and the complete measurement statistics respectively.  { Moreover, we  present lower bounds on the randomness generation rates of chained Bell inequality for different $n$, which are compared with  the family of Bell inequalities $J_{\gamma }$ proposed by Wooltorton et al. \cite{ref6}. Interestingly, we find that chained Bell inequality with three measurement settings has certain advantages and could be used to improve randomness generation rates in practice, although the improvement that we get is modest at a low level of noise. Our results also show that using more measurements is meaningful beyond Bell scenarios with two binary measurement settings.}

 {Note that it is a significant view to show a tight bound by giving an explicit construction to attain the given randomness for the corresponding values of the chained Bell inequality. However, it is technically hard for chained inequality since the number of measurements in chained inequality is much more than in CHSH inequality, and we will leave it for future work.}

 {It would be interesting to further investigate analogous upper bounds in  bipartite scenarios by utilizing our method to extend to other Bell inequalities with many settings such as Gisin Bell inequality \cite{ref26}, including nonlinear Bell inequality based
on products of correlators \cite{ref27}. Finally, it is worth exploring whether our method can be extended to Bell scenarios with many settings and more  parties as well as steering scenarios.}

\section*{\bf Acknowledgments} This work is supported by the Shandong Provincial Natural Science Foundation for Quantum Science ZR2021LLZ002, No.ZR2020LLZ003, the Fundamental Research Funds for the Central Universities No.22CX03005A, the NSFC under Grant Nos. 12075159 and 12175147, Beijing Natural Science Foundation (Z190005), and the Academician Innovation Platform of Hainan Province.

\section*{\bf Data availability statement} All data generated or analyzed during this study are included in this published article.

\end{document}